\documentclass[oribibl,envcountsame,envcountreset,envcountsect]{llncs}
 
\usepackage{amssymb} 
\usepackage{amsmath} 
\usepackage[svgnames]{xcolor}

\usepackage{tikz}
\usetikzlibrary{plotmarks}

\usepackage[round]{natbib}
\usepackage[svgnames]{xcolor}
\usepackage{comment}
\usepackage{xspace}
\usepackage{enumitem}
\usepackage{hyphenat}
\usepackage[normalem]{ulem}
\usepackage{algorithm}
\usepackage{algorithmic}

\usepackage{longtable}
\usepackage{hyperref}

\definecolor{darkgreen}{rgb}{0,.35,0}
\definecolor{darkblue}{rgb}{0,0,.5}
\definecolor{darkred}{rgb}{.6,0,0}

\numberwithin{equation}{section}
\makeatletter
\spn@wtheorem{fact}{Fact}{\bfseries}{\itshape}
\makeatother

\DeclareMathOperator{\lc}{lc}

\newcommand{\ZZ}{{\mathbb{Z}}}
\newcommand{\QQ}{{\mathbb{Q}}}
\newcommand{\RR}{{\mathbb{R}}}

\newcommand{\NN}{{\mathbb{N}}}
\newcommand{\KK}{{\mathbb{K}}}

\newcommand{\diff}{{\lower3pt\hbox{\large$'$}}}

\makeatletter
\@twosidefalse
\makeatother

\makeatletter
\renewcommand*\l@author[2]{}
\renewcommand*\l@title[2]{}
\makeatletter

\begin{document}

\title{Factorization of $\ZZ$-homogeneous Polynomials in the First
  ($q$-)Weyl Algebra}

\author{Albert Heinle${}^1$ and Viktor Levandovskyy${}^2$}

\institute{${}^1$Cheriton School of Computer Science, 200 University
  Ave. West, Waterloo, \mbox{N2L 3G1}, Ontario, Canada. Email: \url{aheinle@uwaterloo.ca}\\
${}^2$Lehrstuhl D f\"ur Mathematik, RWTH Aachen University,
Templergraben 64, 52062 Aachen, Germany. Email: \url{viktor.levandovskyy@math.rwth-aachen.de}}

\maketitle

\begin{abstract}
We present algorithms to factorize weighted
homogeneous elements in the polynomial first Weyl algebra and $q$-Weyl algebra, which
are both viewed as $\ZZ$-graded rings. We show that
factorization of homogeneous polynomials can be almost completely
reduced to commutative univariate factorization over the same base
field with some additional uncomplicated combinatorial
steps. This
allows to deduce the complexity of our algorithms in detail.
Furthermore, we will show for homogeneous polynomials that
irreducibility in the polynomial first Weyl
algebra also implies irreducibility in the rational one, which is of interest
for practical reasons. We report on our implementation in the computer algebra
system \textsc{Singular}. For homogeneous polynomials, it outperforms
 currently available
implementations for factoring in the first Weyl algebra -- in speed as
well as in elegancy of the results.
\end{abstract}


\tableofcontents

\section{Introduction}

Algebras of operators, such as the $q$-Weyl and the Weyl algebras, are important
objects to study since, among other things, one can derive properties of the solution spaces of their associated systems of equations one wants to solve. Especially
concerning the problem of finding the solutions of a linear ordinary
($q$-)differential equation, the preconditioning step of factorizing this operator may come in helpful.

But often algebras of operators are noncommutative polynomial
rings, and a factorization of an element in those algebras is neither unique
in the classical sense (i.e. unique up to multiplication by a unit),
nor easy to compute at all in general.

Nevertheless, a lot has been done in this field in the past. Tsarev
has studied the form, number and the properties of the factors of a
linear differential operator in \cite{Tsarev:1994} and \cite{Tsarev:1996},
where he uses and extends the work presented in \cite{Loewy:1903}
and \cite{Loewy:1906}.

A very general approach to noncommutative algebras and their
properties, including factorization, is also done by Bueso,
Gomez-Torrecillas and Verschoren in (\cite{Bueso:2003}). They provide several algorithms
and introduce various points of views when dealing with noncommutative
polynomial algebras.

In his dissertation van Hoeij developed an algorithm to factorize
a linear differential operator (\cite{Hoeij:1996}). There were several
papers following that dissertation using and extending those
techniques (e.g. \cite{Hoeij:1997}, \cite{van1997formal} and \cite{Hoeij:2010}), and
nowadays this algorithm is implemented in the \texttt{DETools} package
of \textsc{Maple} (\cite{Maple}) as the standard algorithm for
factorization of those operators.

For the finite field case, Giesbrecht and Zhang have developed a
polynomial time algorithm to factor polynomials in
$\mathbb{F}_q(t)[\mathcal{D};\sigma, \delta]$
(\cite{giesbrecht2003factoring}). This includes the Weyl algebras
with rational function coefficients over a finite field. The applied
methodology extends the results in \cite{giesbrecht1998factoring}.

From a more algebraic point of view and dealing only with strictly
polynomial noncommutative algebras, i.e. all units are in the center
of the algebra,
Melenk and Apel developed a package for the
computer algebra system \textsc{REDUCE} (\cite{Melenk:1994}). This package
provides tools to deal with noncommutative polynomial algebras and
also contains a factorization algorithm for the supported
algebras.

In the computer algebra system \textsc{ALLTYPES}
(\cite{Schwarz:2009}), which is based on \textsc{REDUCE} and solely accessible as a
web-service, Schwarz and Grigoriev have
implemented the algorithm for factoring differential operators they
introduced in \cite{GrigorievSchwartz:2004}.

Beals and Kartashova (\cite{beals2005constructively}) consider
the problem of finding a first-order left hand factor of an element
from the second Weyl algebra over a computable differential field, where
they are able to deduce parametric factors. Similarly,
Shemyakova studied factorization
properties of linear partial differential operators in \cite{parameters:shemyakova:2007},
\cite{shemyakova2009multiple} and
\cite{2010:shemyakova:refinement}.

Concerning special classes of polynomials in algebras of operators, the
paper \cite{foupouagnigni2004factorization} deals with
factorization of fourth-order differential equations
satisfied by certain Laguerre-Hahn orthogonal polynomials (\cite{nikiforov1988special}).

Those algorithms and implementations are very well written and
they are able to factorize a large number of polynomials we give them
as input. Nonetheless, as we will see in this paper, there exists a large
class of polynomials that seem to form the worst case for the mentioned algorithms. One can use a different approach to obtain a
factorization of such polynomials very quickly, and we will prove that this factorization into irreducible elements is also irreducible in the rational first
($q$-)Weyl algebra. This approach extends the one developed in
\cite{Heinle:2010}. In this work we deal with this class of polynomials by describing our
methods in detail and providing a complexity estimate for the
factorization in the case, where the underlying field is
computable. A very recent algorithm for factoring general polynomials, which is  based on the results presented here, is given in
\cite{giesbrecht2014factoring}.
 We state another main result in Theorem
\ref{thrm:irreducibleRationalWeyl}. There, we prove that irreducible
homogeneous polynomials in the polynomial first Weyl algebra stay
irreducible when considering them as elements in the rational first
Weyl algebra. This is rather unexpected, as
this statement is not true for general, i.e. inhomogeneous, polynomials.

Our algorithms are implemented in the computer algebra system 
\textsc{Singular} (\cite{Singular:2012}, \cite{GreuelPfister:2007}, \cite{Plural:2010}), and since version 3-1-3 they became part of
the distribution as the library \texttt{ncfactor.lib}.

\subsection{Preliminaries}

We will start by introducing the first $q$-Weyl algebra and the
first Weyl algebra. By $\KK$, we always denote an arbitrary field. All
algebras are unital associative $\KK$-algebras. For
the complexity discussions, we assume that
\begin{itemize}
\item[(i)] $\KK$ is computable and
its arithmetics have polynomial costs with respect to the bit-size of the
elements in $\KK$.
\item[(ii)] There exists a norm $|\cdot|: \KK \to \RR$. The
  representation size in bits for an element $k \in \KK$ is bounded by
  $\lceil \log |k| \rceil$.
\end{itemize}
The role of the invertible parameter $q$ can be different: from $q \in \KK$ to
$q$ being transcendental over $\KK$. We use the unified notation $\KK(q)$ for
all these cases. Moreover, for $m\in\NN$ we denote by $\underline{m}$ the set 
$\{1, \ldots, m\}$.

\begin{definition}
The \textbf{polynomial first $q$-Weyl algebra} $Q_{1}$ is
defined as
\begin{eqnarray*}
        Q_1 := \KK(q)\langle x,\partial | \partial x = qx\partial +1\rangle.
\end{eqnarray*}
For the special case where $q = 1$ we have the \textbf{polynomial first Weyl algebra}, which is denoted by $A_1$.
\end{definition}

\begin{remark}
 The  first $q$-Weyl algebra can be viewed as an algebra associated to the
 operator $$\partial_{q}: f(x) \mapsto \frac{f(qx)-f(x)}{(q-1)x},$$
also known as the $q$-derivative, where $f$ is a univariate function
in $x$ (cf. \cite{Cheung:2002}).

For $q = 1$, the operator is still well defined. This can be seen in
the following way. Let $f = \sum_{i = 0}^n a_ix^i$, where $n \in
\NN_0$ and $a_i \in \KK$. Then
$$f(qx) - f(x) = \sum_{i=0}^na_i(qx)^i - \sum_{i = 0}^n a_ix^i =
\sum_{i = 0}^n a_ix^i(q^i-1).$$
The expression $q-1$ is clearly a divisor of $q^i-1$ for all $i \geq 1$, and we obtain
$$\frac{f(qx)-f(x)}{(q-1)x} = \sum_{i=1}^na_ix^{i-1}\left(\sum_{j=0}^{i-1}q^j\right).$$

\end{remark}


The first ($q$-)Weyl algebra possesses a nontrivial 
$\ZZ$-grading -- introduced by M.~Kashiwara and B.~Malgrange in a
broader context of the so-called $V$-filtration in 1983
  (\cite{Kashiwara:1983}, \cite{Malgrange:1983}) -- using the weight
  vector $[-v,v]$ for non-zero $v \in \ZZ$ on
the tuple $[x,\partial]$. For
simplicity, we will choose $v := 1$. In what follows, $\deg$ denotes
the degree induced by this weight vector. We will write $\deg_x$
and $\deg_\partial$ for the degree of a polynomial in $Q_1$
resp. $A_1$ with respect to $x$ and $\partial$. 
From now on, we mean by \emph{homogeneous} or \emph{graded} a polynomial, which is homogeneous with respect to the weight vector $[-1,1]$.

\begin{example}
  We have $\deg(x\partial) = \deg(\partial x + 1) = 0$.
  Another homogeneous polynomial is
  $$x\partial^2 + x^4\partial^5 + \partial = (x\partial +
  x^4\partial^4 + 1)\partial,$$
  which is of degree one.
\end{example}

For $n \in \ZZ$, the $n$th graded part (cf. \ref{prop:A0completeWeylDescription} for more detailed description) of $Q_1$ and analogously the $n$th graded part of $A_1$ is given by
$$Q_1^{(n)} := \left\{\sum_{j-i=n} r_{i,j}x^i\partial^j | i,j \in \NN_0,
r_{i,j} \in \KK\right\},$$

i.e. the degree of a monomial is determined by the difference of its
powers in $x$ and $\partial$.

Concerning this choice of degree, the so called \textbf{Euler operator}
$$\theta := x\partial,$$
which is homogeneous of degree 0, will play an important role as we will see soon.

First of all, let us investigate some commutation rules the Euler operator
has with $x$ and $\partial$. For $Q_1$, in order to abbreviate the size
of our formulas, we introduce the so called $q$-bracket.

\begin{definition}
  For $n \in \NN$, we define the \textbf{$q$-bracket}
  $[n]_q$ by
$$[n]_q := \frac{1-q^n}{1-q} = \sum_{i = 0}^{n-1}q^i.$$
\end{definition}

\begin{lemma}[Compare with \cite{SaStuTaka:2000}]
  \label{lem:rewriteKTheta}
  In $A_1$, the following commutation rules do hold for $n \in \NN$:
  \begin{eqnarray*}
    \theta x^n &=& x^n (\theta + n)\\
    \theta \partial^n & = &\partial^n(\theta - n).
  \end{eqnarray*}
  More generally, in $Q_1$ the following commutation rules
  do hold for $n \in \NN$:
  \begin{eqnarray*}
    \theta x^n &=& x^n (q^n\theta + [n]_q)\\
    \theta \partial^n & = & \frac{\partial^n}{q}\left(
      \frac{\theta-1}{q^{n-1}} - \frac{q^{-n+2}-q}{1-q} \right).
  \end{eqnarray*}
\end{lemma}

Those rules follow via induction on $n \in \NN$. 

\begin{remark}
  If the characteristic of $\KK$ is some prime number $p$,
  the elements $x^{ap}$ (resp. $\partial^{ap}$) for all $a \in \NN_0$
  commute with $\theta$ in $A_1$.
\end{remark}

\begin{remark}
  With the help of the Lemma above one can also easily see that the so
  called polynomial first shift algebra
  $$\KK \langle n,s | sn = (n+1)s \rangle$$
  is a subalgebra of the first Weyl algebra $A_1$. An embedding of
  a polynomial $p = \sum_{i = 0}^{n} p_i(n)s^n$ from the shift algebra,  where $p_i \in
  \KK[n]$, into the first Weyl algebra is done via
  the following homomorphism of $\KK$-algebras:
  $$\iota: \KK \langle n,s | sn = (n+1)s \rangle \to A_1, \quad \sum_{i =
    0}^{n} p_i(n)s^n \mapsto \sum_{i = 0}^{n} p_i(\theta)\partial^n.$$

Therefore, the
  factorization techniques developed here can also be applied to
  the first shift algebra.
\end{remark}

The commutation rules in Lemma \ref{lem:rewriteKTheta} can of course be extended to arbitrary
polynomials in $\theta$.

\begin{corollary}
\label{cor:thetaswap}
Consider $f(\theta):=f \in \KK[\theta], \theta := x\partial$. Then,
in $Q_1$, for
all $n \in \NN$ we have 
\begin{eqnarray*}
  f(\theta)x^n &=& x^nf(q^n\theta + [n]_q),\\
  f(\theta)\partial^n &=& \partial^nf\left(\frac{1}{q}\left(
      \frac{\theta-1}{q^{n-1}} - \frac{q^{-n+2}-q}{1-q}\right)\right),
\end{eqnarray*}
whereas in $A_1$ we have
\begin{eqnarray*}
  f(\theta)x^n &=& x^nf(\theta + n),\\
  f(\theta)\partial^n &=& \partial^nf\left( \theta -n\right).
\end{eqnarray*}
\end{corollary}

Those are the basic tools we need to explain our approach for
factoring homogeneous polynomials in the first Weyl and the first
$q$-Weyl algebra.

For the complexity discussion, let us define some
constants we will utilize in order to estimate the operations needed
to perform our methods.

\begin{definition}
  \label{def:complexConstants}

  Let us denote by
  $\omega_q(n,c)$, for $n,c \in \NN_0$,
  the number of bit operations that an algorithm for factoring a polynomial of
  degree $n$ in a univariate polynomial ring over $\KK(q)$, where each
  coefficient has at most bit-size $c$, needs to
  perform.

  We denote for $n,c \in \NN_0$ by
  $\rho_q(n,c)$
  the number of bit operations needed to multiply 
  two polynomials
  in a univariate polynomial ring over $\KK(q)$,  where each
  polynomial has degree at most $n$ and where $c$ is the
  maximal bit size of each coefficient in the two polynomials.

  We will write $\mathcal{S}_q(n,k,c,\sigma)$, $n,c\in
  \NN_0, k \in \ZZ$, $\sigma \in \mathrm{Aut}(\KK[x])$, for the number of
  bit operations needed for
  computing $f(\sigma^k(x))$ for a polynomial $f$ in $\KK(q)[x]$ of degree $n$, where $x$ is an
  indeterminate and transcendental over $\KK(q)$ and each coefficient
  of $f$ has at most bit-size $c$.  

  

  If we deal with the case $q=1$, we will omit writing the subscript.
\end{definition}

For a detailed complexity discussion, we need to specify the expected
output of our factorization algorithms.

\begin{definition}
  \label{def:factorizationNormalize}
  Let $A$ be a polynomial algebra over a field $\KK$ and $f\in
  A\setminus \KK$ be a 
  polynomial.  For a fixed totally ordered monomial $\KK$-basis of $A$,
  the leading coefficient $\lc(f)$ of $f$ is uniquely defined. A {\bf
    nontrivial factorization} of $f$ is a tuple $(c,f_1,\ldots,f_m)$,
  where $c \in \KK\setminus\{0\}$, $f_1,\ldots,f_m \in
  A\setminus\{1\}$ are \textbf{monic} (i.e. they satisfy $\lc(f_i)=1$)
  and $f= c \cdot f_1 \cdots f_m.$

  By a slight abuse of notation, we may omit the first element in the
  tuple if $c=1$.
\end{definition}





The following lemma will provide a complexity estimate of the cost 
of testing whether a polynomial in $Q_1$ resp. $A_1$ is homogeneous.

\begin{lemma}
  \label{lem:cheapHomogTest}
  In order to determine whether a polynomial $p \in Q_1$ resp. $p \in
  A_1$ is homogeneous, it requires $\#\{\text{Terms in } p\}$ integer
  additions and comparisons.
\end{lemma}
\begin{proof}
  A polynomial $p$ is homogeneous with respect to our definition if
  and only if in every term the difference between the degree in $x$
  and the degree in $\partial$ is the same. Hence our statement follows.
\end{proof}

Graded elements 
enjoy numerous nice properties, in particular regarding factorizations.
\begin{lemma}
\label{GradedFctr}
Let $(\Gamma,+)$ be a monoid, totally ordered by $<$, such
that $a < b \Rightarrow a+c < b+c$ for all $a,b,c\in \Gamma$.
Moreover, let $D$ be a domain over a field $\KK$, nontrivially graded by $\Gamma$, that is $D = \oplus_{\gamma \in \Gamma} D_{\gamma}$ for $\KK$-vector spaces $D_{\gamma}$ and $D_{\alpha} \cdot D_{\beta} \subseteq D_{\alpha+\beta}$ holds $\forall \alpha, \beta \in \Gamma$.

Consider $d \in D\setminus\{0\}$.  
If there is $m\geq 1$ and $d_i \in D$, such that 
$d = d_1 \cdot \ldots \cdot d_m$, then 
$d$ is $\Gamma$-graded 
if and only if
$d_1, \ldots, d_m$ are $\Gamma$-graded.
\end{lemma}
\begin{proof}
The $\Leftarrow$ direction follows by the definition of grading, so it remains to
prove the $\Rightarrow$ direction. For an element $f \in D\setminus\{0\}$, let us denote by $\alpha(f) \in \Gamma$ resp. by $\omega(f) \in \Gamma$ the degree of the highest resp. the lowest nonzero graded part of $f$. Note, that $\omega(f) \leq \alpha(f)$. Thus $f = f_{\alpha(f)} + \ldots + f_{\omega(f)}$ and, moreover, $f$ is graded if and only if $f = f_{\alpha(f)} = f_{\omega(f)}$. 

Suppose $d=bc$, where $b=b_{\alpha(b)} + \ldots + b_{\omega(b)}$ and 
$c=c_{\alpha(c)} + \ldots + c_{\omega(c)}$. 
Then $bc = b_{\alpha(b)} c_{\alpha(c)} + \ldots +b_{\omega(b)}c_{\omega(c)}$ is the graded decomposition of $d=bc$,
and $(bc)_{\alpha(bc)} = (bc)_{\alpha(b) + \alpha(c)} = b_{\alpha(b)} c_{\alpha(c)}$ since $D$ is a domain. Analogously $(bc)_{\omega(bc)} =b_{\omega(b)}c_{\omega(c)}$. Since $d=bc$ is graded one has thus
$\alpha(bc)=\omega(bc)$, that is $\alpha(b)+ \alpha(c) = \omega(b) + \omega(c)$. Together with $\alpha(b) \geq \omega(b), \alpha(c)  \geq \omega(c)$ this delivers $\alpha(b) = \omega(b)$ and $\alpha(c) = \omega(c)$,
proving the claim.
\end{proof}




\section{A New Approach for Factoring Homogeneous Polynomials in the First ($q$-)Weyl Algebra}

The main idea of our factorization technique lies in the reduction to
a commutative univariate polynomial subring of $A_1$ resp. $Q_1$,
namely $\KK[\theta]$.
We will show that there are only two monic irreducible elements in $\KK[\theta]$, that are
reducible in $A_1$ resp. $Q_1$. Hence, factoring graded elements in $A_1$ (which have a representation as
  $\KK[\theta]$ modules) can be
reduced to factoring in $\KK[\theta]$, identifying these two
elements in a given list of factors, and interchanging using
commutation rules.

We will start with discussing how to find one factorization of a
given homogeneous polynomial, which, in the process, also leads us to the answer of
the question how to find all possible factorizations.



\subsection{Factoring homogeneous polynomials of degree zero}
\label{sbsctn:facDegZero}

The following lemma shows that we can rewrite every homogeneous
polynomial of degree zero in $A_1$ resp. $Q_1$ as a polynomial in $\KK[\theta]$.

\begin{lemma}[Compare with \cite{SaStuTaka:2000}, Lemma 1.3.1]
\label{lem:homogToTheta}
In $A_1$, we have the following identity for $n \in \NN$:
$$x^n\partial^n = \prod_{i = 0}^{n-1} (\theta - i).$$
In $Q_1$, one can rewrite $x^n\partial^n$ as element in $\KK[\theta]$
and it is equal to
$$\frac{1}{q^{T_{n-1}}}\prod_{i=0}^{n-1}\left(\theta -
 \sum_{j=0}^{i-1} q^j\right) = \frac{1}{q^{T_{n-1}}}\prod_{i=0}^{n-1}\left(\theta
 - [i]_q\right),$$
where $T_i$ denotes the $i$th triangular number
$\sum_{j=0}^i j = \frac{i(i+1)}{2}$
for all $i \in \NN_0$.
\end{lemma}

Therefore the factorization of a homogeneous polynomial $p$ of degree zero
can be done by rewriting $p$ as element in $\KK[\theta]$ and
factor it in $\KK[\theta]$, which is for practical choices of $\KK$ well
implemented in every computer algebra system.

Of course, this would not be a complete factorization, as there are
still elements irreducible in $\KK[\theta]$, but reducible in $Q_1$
resp. $A_1$. An obvious example is $\theta$ itself. Fortunately, there are only two monic polynomials irreducible in
$\KK[\theta]$, but reducible in $A_1$ resp. $Q_1$. This is shown by
Lemma \ref{lem:thetairred}, which requires the following proposition
for its proof.

\begin{proposition}
  \label{prop:A0completeWeylDescription}
  $Q_1^{(0)}$ resp. $A_1^{(0)}$ is a $\KK$-algebra, generated by
  the element $\theta := x\partial$. The graded direct summands
  $Q_1^{(k)}$ resp. $A_1^{(k)}$ are cyclic $Q_1^{(0)}$
  resp. $A_1^{(0)}$ bi-modules generated by the
  element $x^{-k}$, if $k<0$, or by $\partial^k$, if $k>0$.
\end{proposition}
\begin{proof}
  The first statement can be seen using Lemma \ref{lem:homogToTheta},
  as we can identify $Q_1^{(0)}$ resp. $A_1^{(0)}$ with $\KK[\theta]$.

  For the second statement recall that being homogeneous of degree $k
  \in \ZZ$ for a polynomial $p \in
  Q_1^{(k)}$ resp. $p \in A_1^{(k)}$  means, that every monomial is -- for a certain $n \in \NN_0$ -- of the form
  $x^n\partial^{n+k}$, if $k \geq 0$, or of the form
  $x^{n-k}\partial^n$, if $k < 0$. Since we can transform $x^n\partial^n$
  into an expression in $\KK[\theta]$ via Lemma \ref{lem:homogToTheta}
  and use the commutation rules in Lemma \ref{lem:rewriteKTheta}, we
  can move $x^{-k}$ resp. $\partial^k$ to the right and the
  left and hence obtain the
  desired bi-module structure.
\end{proof}

\begin{lemma}
\label{lem:thetairred}
The polynomials $\theta$ and $\theta + \frac{1}{q}$ are the only
irreducible monic elements in $\KK[\theta]$ that are reducible in
$Q_1$.
For $A_1$, the polynomials $\theta $ and $\theta +1$ are the only
irreducible monic elements in $\KK[\theta]$ that are reducible in $A_1$.
\end{lemma}
\begin{proof}
We will only consider the proof for $Q_1$, as the proof for $A_1$ is
done in an analogue way. Let $f \in \KK[\theta]$ be a monic polynomial. Assume that it is
irreducible in $\KK[\theta]$, but reducible in $Q_1$. Let
$\varphi,\psi$ be elements in $Q_1$ with $\varphi\psi = f$. Then
$\varphi$ and $\psi$ are homogeneous and $\varphi \in Q_1^{(-k)}, \psi
\in Q_1^{(k)}$ for a $k\in \ZZ\setminus\{0\}$. As for the case where $k$
is negative a similar argument is applicable, we assume without loss of
generality that $k$ is positive.

Due to Proposition \ref{prop:A0completeWeylDescription}, we have for some $\tilde \varphi, \tilde \psi \in \KK[\theta]$
$$
  \varphi = \tilde \varphi(\theta) x^k, \qquad \psi =  \tilde \psi (\theta) \partial^k.
$$

\noindent
Using Corollary \ref{cor:thetaswap}, we obtain
$$f = \tilde \varphi (\theta) x^k \tilde \psi (\theta) \partial^k = \tilde
\varphi (\theta) x^k\partial^k\tilde \psi \left(\frac{1}{q}\left(
      \frac{\theta-1}{q^{n-1}} - \frac{q^{-n+2}-q}{1-q}\right)\right).$$

\noindent
As we know from Lemma \ref{lem:rewriteKTheta} the equation
$$x^k\partial^k =\frac{1}{q^{T_{k-1}}}\prod_{i=0}^{k-1} \Bigl( \theta -
 \sum_{j=0}^{i-1} q^j\Bigr)$$
holds.

Thus, because we assumed $f$ to be irreducible in $\KK[\theta]$, we
must have $\tilde \varphi,\tilde \psi \in \KK$ and $k = 1$ due to
Lemma \ref{lem:rewriteKTheta}. Because
$f$ is monic, we must also have $\tilde\varphi = \tilde\psi^{-1}$.

As a result, the only possible $f$ is $f=\theta$. If we
originally had chosen $k$ to be negative, the only possibility
for $f$ would be
$f = \theta + \frac{1}{q}$. This completes the proof.
\end{proof}

Therefore, we have a procedure
for factoring a homogeneous polynomial $p \in A_1$ (resp. $p \in Q_1$) of
degree zero in $\KK[\theta]$. It is done using the following steps.

\begin{enumerate}[topsep=2pt,itemsep=0pt,parsep=2pt]
  \item Rewrite $p$ as an element in $\KK[\theta]$;
  \item Factorize $p$ in $\KK[\theta]$ using
    commutative methods, i.e. obtain a list $[c,p_1,\ldots, p_\ell] \in
    \KK[\theta]^\ell$, $\ell \in \NN$, $c\in\KK$, where $c\cdot p_1\cdots p_\ell = p$.
  \item For every $p_j$, $j \in \{1,\ldots, \ell\}$, that is equal to
    $\theta$ or $\theta+1$ (resp. $\theta + \frac{1}{q}$),
    remove $p_j$ from the list and insert into position $j$ and $j+1$
    the elements
    $x_i, \partial_i$ resp. $\partial_i, x_i$.
  \item  Replace for every element in the list from the previous step
    $\theta$ by $x\cdot \partial$. Return
    the resulting list.
\end{enumerate}

  Let us consider the complexity of the above steps to factor a
  homogeneous element of degree zero in $A_1$.

  \textbf{Ad step 1:} The polynomial $p$ has, due to the assumption of
  being homogeneous of degree zero, the form
  \begin{align}
  \label{eq:complexityP}
  p = \sum_{i = 0}^{n}p_ix^i\partial^i, n \in \NN, p_i \in \KK \text{
    (resp. } \KK(q) \text{)}.
  \end{align}
  In order to transform it into an element in $\KK[\theta]$, we have to
  apply the rewriting rule stated in Lemma \ref{lem:homogToTheta} for every term
  $x^i\partial^i$ in $p$. For that, one makes use of the identity
$$x^{n+1}\partial^{n+1} = x^n\partial^n \cdot (\theta - n).$$
  
  Thus, in order to perform step 1, we need to perform for every $i \in \underline{n}$ 
  a multiplication of a polynomial in $\KK[\theta]$
  of degree $i$ with a polynomial of degree $1$.

  \textbf{Ad step 2:} Unfortunately, the factorization problem even in
  the univariate case does not have polynomial complexity in
  general. One might face exponential complexity with respect to the
  bit-length of the coefficients in $\KK$ or it might even be undecidable, depending on the choice of
  $\KK$.

  An example for a polynomial-time complexity with respect to the
  bit-length of the coefficients would be $\KK=\QQ$, due to the famous LLL algorithm by Lenstra, Lenstra,
  Lovasz developed in 1982 (\cite{LLL:1982}). For certain classes of
  fields, including algebraic ones, polynomial time algorithms have
  been discovered in \cite{chistov1986algorithm} and
  \cite{grigor1984factoring}. For further
  readings on the complexity of the factorization problem we also recommend
  \cite{Kaltofen:1982} and \cite{von2013modern}. As given in Definition
  \ref{def:complexConstants}, we simply write $\omega(n)$
  resp. $\omega_q(n)$ for the amount of bit operations needed for
  factoring a univariate polynomial of degree $n$.

  \textbf{Ad step 3:} In order to find and identify the polynomials,
  it does not require any operations on the polynomials other than
  comparisons.

  \textbf{Ad step 4:} For each monomial in each factor that has degree
  zero, we need to
  replace $\theta$ by $x\cdot \partial$ and bring it into normal form,
  i.e. each monomial in the end must have the form $x^i\partial^i$ for $i \leq n$.
  This can be calculated, up to a constant factor, with the same number
  of operations as performed for step 1, since we only need to reverse
  the mapping outlined there.

  Thus, we can formulate the following corollary.

  \begin{corollary}
    \label{cor:complexityHomogZero}
    Given $p$ as in (\ref{eq:complexityP}), and let $b$ be the maximal
    coefficient in $p$ with respect to its bit-size. In order to obtain one factorization  of $p$ over $Q_1$, it requires
    \begin{align}
      \label{eq:oneFactorCost}
      O\left(n\cdot \rho_q(n,\lceil \log| n! | \rceil)+
      \omega_q(n,\lceil \log |b \cdot n!| \rceil )\right)
    \end{align}
    bit operations.
  \end{corollary}

\begin{example}
  Let $\KK:=\QQ$ and 
  $$p := x^3\partial^3+4x^2\partial^2+3x\partial \in A_1.$$
  Clearly $p$ is homogeneous of degree zero; rewritten in $\KK[\theta]$, one obtains
  $$p = \theta^3 + \theta^2 + \theta.$$
  This polynomial factorizes in $\KK[\theta]$ to $\theta \cdot
  (\theta^2 + \theta + 1)$, which further factorizes as $\theta$ is reducible to
  $x \cdot \partial \cdot (\theta^2 + \theta + 1) \in A_1$.
  To get more (in fact, as we will see in the next subsection, all) possible factorizations of $p$, we apply the commutation
  rules with $x$ resp. $\partial$ and obtain the following other
  factorizations:
  \begin{eqnarray*}
    (\theta^2 + \theta + 1) \cdot x \cdot \partial,\\
    x \cdot (\theta^2 + 3\theta + 3) \cdot \partial.
  \end{eqnarray*}
\end{example}



\subsection{Factoring homogeneous polynomials of arbitrary degree}
\label{sbsctn:facDegArb}

Fortunately, the hard work is already done and factoring of
homogeneous polynomials of arbitrary degree is just a small further
step.

The reason is Proposition \ref{prop:A0completeWeylDescription}, which
leads to the following steps to obtain one factorization of a homogeneous polynomial $p \in
  Q_1^{(k)}$ resp. $p \in A_1^{(k)}$ of degree $k \in \ZZ$.
\begin{enumerate}
  \item Represent $p$ as $\tilde p x^{-k}$ resp. $\tilde
    p \partial^{k}$, where $\tilde p$ in $A_1^{(0)}$, written as
    polynomial in $\KK[\theta]$. We need
    $O(d^2\cdot \rho_q(d,\lceil \log |b\cdot d!|\rceil))$, where $d := \min\{\deg_x(p),
    \deg_\partial (p)\}$ and $b\in \KK$ denotes the maximal
    coefficient in $p$ with respect to the bit-size, operations to obtain this $\tilde p$.
    Afterwards, if $k<0$, one additional application of a $k$-shift to $\tilde
    p$ is required.
  \item Factorize $\tilde p$ -- which is homogeneous of
    degree zero -- using the
    steps shown in the previous subsection.
\end{enumerate}

Now we have everything we need to formulate an algorithm to find one
factorization of a homogeneous element in $A_1$ resp. $Q_1$, namely
Algorithm \ref{alg:homogfac} which can be found below.

  The
  next corollary states a complexity estimate Algorithm \ref{alg:homogfac}. The
  proof is straightforward and left to the reader.

  \begin{corollary}
    \label{cor:complOneFact}
    Let $p \in Q_1$ be homogeneous of degree $k\in \ZZ$, and let all
    the coefficients in $p$ have at most bit size $b \in \NN_0$. Then,
    due to Proposition \ref{prop:A0completeWeylDescription}, $p$ can be
    written in the form $p = p_0\varphi^{|k|}$, where $p_0$ is a
    polynomial of degree $n \in \NN_0$ in $\KK(q)[\theta]$ and $\varphi \in \{x,\partial\}$. Obtaining
    one factorization in $Q_1$ of $p$ requires
    $$O\left(n\cdot \rho_q(n,\lceil \log |n!| \rceil) +
    \omega_q(n,\lceil \log |b\cdot n!|)
    + \mathcal{S}_q(n,k,\lceil \log | b\cdot n!|\rceil, \sigma)\right)$$
    bit operations, where $\sigma(x) = x+1$ if $q=1$, and $\sigma(x) =
    q\cdot x
    +1$ otherwise.
\end{corollary}

\begin{algorithm}
	\caption{HomogFac: Factorization of a homogeneous polynomial
          in the first ($q$-)Weyl algebra}
\label{alg:homogfac}
        \begin{flushleft}
          \textit{Input:} $h \in A_1^{(m)}$ (resp. $h \in Q_1^{(m)})$, where $m \in \mathbb{Z}$ \\
          \textit{Output:} $(f_1,\ldots,f_n) \in A_1^n $
          resp. $(f_1,\ldots,f_n) \in Q_1^n$, such that $f_1 \cdot \ldots \cdot f_n = h$, $n \in \mathbb{N}$\\
          \textit{Assumption:} $h$ is normalized, i.e. the leading
          coefficient is 1. \\
        \end{flushleft}
	\begin{algorithmic}[1]
		\IF{$m \neq 0$}
			\IF{$m<0$}
				\STATE Get $\hat h \in A_1^{(0)}$ such that $h = \hat h x^{-m}$
				\STATE \textit{factor} $:=$ $(\underbrace{x,\ldots,x}_{-m \text{ times}})$
			\ELSE
				\STATE Get $\hat h$ such that $h = \hat h \partial^m$
				\STATE \textit{factor} $:=$ $(\underbrace{\partial,\ldots,\partial}_{m \text{ times}})$
			\ENDIF
		\ELSE
			\STATE $\hat h := h$
			\STATE \textit{factor} $:=$ $1$
		\ENDIF
		\STATE $(\hat f_1,\ldots, \hat f_l):=$ Factorization
                of $\hat h$ as element in $\KK[\theta]$ ($l \in \NN$)
		\STATE $(\hat{\hat{f_1}},\ldots,\hat{\hat{f_l}}) := $ Substitute $\theta$ by
                $x \cdot \partial$ in $(\hat f_1,\ldots,\hat f_n)$
                \STATE $\textit{result} := ()$
		\FOR{$i$ from 1 to $l$}\label{ln:beginDecomposeTheta}
                  \IF{$\hat{\hat{f_i}} = x \cdot \partial$}
                    \STATE Append $x$ and $\partial$ to \textit{result}
                  \ELSE
                    \IF{$\hat{\hat{f_i}} = \partial \cdot x$}
                       \STATE Append $\partial$ and $x$ to \textit{result}
                    \ELSE
                       \STATE Append $\hat{\hat{f_i}}$ to
                       \textit{result}
                     \ENDIF
                  \ENDIF
                \ENDFOR \label{ln:endDecomposeTheta}
                \STATE Append each element in \textit{factor} to \textit{result}
		\RETURN{\textit{result}}
	\end{algorithmic}
\end{algorithm}

We also would like to address the topic how to obtain all possible
factorizations of a homogeneous polynomial. As mentioned before, the factorization of a polynomial in a
noncommutative ring is generally not unique in the classical sense, i.e. up to multiplication by units or
up to interchanging factors. Thus several different factorizations can
occur. For the homogeneous case, they can fortunately be easily
characterized by the commutation rules from Lemma
\ref{lem:rewriteKTheta} and the identities from Lemma \ref{lem:thetairred}.
This is proven by the following Lemma.

\begin{lemma}
Let $z \in \ZZ$ and $p \in A_1^{(z)}$, resp. $p \in Q_1^{(z)}$, is
  monic.  Suppose, that one factorization of $p$ has been constructed
  following Proposition \ref{prop:A0completeWeylDescription} and has
  the form $Q(\theta) \cdot T(\theta) \cdot \psi^{|z|}$, where
  \begin{itemize}
    \item $T(\theta) = (x \partial)^{t} (\partial x)^{s}$, $t,s \in \NN_0$,
  is a product of irreducible factors in $\KK[\theta]$, which are
  reducible in $A_1$, resp. $Q_1$,
    \item $Q(\theta)$ is the product of irreducible factors
  in both $\KK[\theta]$ and $A_1$ (resp. $Q_1)$, and
    \item $\psi = x$, if $z<0$, and $\psi = \partial$ otherwise.
  \end{itemize}
  
  Let $p_1 \cdots p_m$ for $m\in \NN$ be another nontrivial
  factorization of $p$. Then this factorization can be derived from
  $Q(\theta) \cdot T(\theta) \cdot \psi^{|z|}$ by using two
  operations, namely (i) ``swapping'', that is interchanging two
  adjacent factors according to the commutation rules and (ii)
  ``rewriting'' of occurring $\theta$ resp. $\theta +1$ 
    ($\theta + \frac{1}{q}$ in the $q$-Weyl case) by
    $x \cdot \partial$ resp. $\partial \cdot  x$.
\end{lemma}
\vspace*{-7pt}
\begin{proof}
Since $p$ is homogeneous, all $p_i$ for $i \in \underline{m}$ are
 homogeneous. Thus 
 each of them can be written in the form
 $p_i = \tilde p_i(\theta)\cdot \psi_{e_i}$, where
 $e_i \in \ZZ$, and $\psi_{e_i} = x^{-e_i}$, if $e_i<0$ and $\psi_{e_i} = \partial^{e_i}$ otherwise. With respect to the commutation
 rules as stated in Corollary \ref{cor:thetaswap}, we can swap the
 $\tilde p_i(\theta)$ to the left for any $2 \leq i \leq m$. 
 Note that it is possible for
 them to be transformed to the form $\theta$ resp. $\theta + 1$
 ($\theta + \frac{1}{q}$ in the $q$-Weyl case), after performing these swapping steps.  I.e., we have
 commuting factors, both belonging to $Q(\theta)$, as well as to
 $T(\theta)$ at the left. Our resulting product is thus $\tilde
 Q(\theta) \tilde T(\theta) \prod_{j=1}^m \psi_{e_j}$, where the factors in $\tilde Q(\theta)$,
 resp. $\tilde T(\theta)$, contain a subset of the factors of
 $Q(\theta)$ resp. $T(\theta)$. By our assumption of $p$ having
 degree $z$, we are able to swap $\psi_z$ to the right in $F
 :=\prod_{j=1}^m \psi_{e_j}$, i.e., $F = \tilde F
 \psi_z$ for $\tilde F \in A_1^{(0)}$. This step may involve
 combining $x$ and $\partial$ to $\theta$ resp. $\theta
 + 1$ ($\theta + \frac{1}{q}$ in the $q$-Weyl case). Afterwards, this is also done to the
 remaining factors in $\tilde F$ that are not yet polynomials in
 $\KK[\theta]$ using the swapping operation. These polynomials are
 the remaining factors that belong to $Q(\theta)$, resp. $T(\theta)$,
 and can be swapped commutatively to their respective positions. Since
 reverse engineering of those steps is possible, we can derive the
 factorization $p_1 \cdots p_m$ from $Q(\theta) \cdot T(\theta)
 \cdot \psi_{z}$ as claimed.
\end{proof}

With the help of the above lemma, we are also able to formulate an
algorithm to find all factorizations of a given homogeneous polynomial
in $A_1$, namely Algorithm \ref{alg:homogfac_all} as stated below.

\begin{algorithm}
	\caption{HomogFacAll: All factorizations of a homogeneous
          polynomial in the first ($q$-)Weyl algebra}
\label{alg:homogfac_all}
        \begin{flushleft}
          \textit{Input:} $h \in A_1^{(m)}$ (resp. $h \in Q_1^{(m)}$), where $ m \in \mathbb{Z}$ \\
          \textit{Output:} $\{(f_1,\ldots,f_n) \in A_1^n \mid f_1 \cdot \ldots \cdot f_n = h, n \in \mathbb{N}\}$\\
          \textit{Assumption:} $h$ is normalized, i.e. the leading
          coefficient is 1.\\
        \end{flushleft}
	\begin{algorithmic}[1]
          \STATE $(f_1,\ldots,f_\nu, g,\ldots,g)$ := HomogFac($h$)
          without lines \ref{ln:beginDecomposeTheta} -- \ref{ln:endDecomposeTheta}\\
          \COMMENT{$\nu \in \mathbb{N}_0, g \in \{x,\partial\}, f_i \in A_1^{(0)}$}
          \STATE Rewrite each $f_i$ as element in $\mathbb{K}[\theta]$
          \STATE \textit{result} $:=$ $\{$Permutations of $(f_1,\ldots,f_\nu,g,\ldots,g)$ with respect to the commutation rules$\}$
          \FOR{$(g_1,\ldots,g_n) \in $ \textit{result}}
            \FOR{$i$ from $1$ to $n$}
               \IF{$g_i = \theta$}
                 \STATE $g_i := x,\partial$
                 \STATE \textit{leftpart} $:=$ $\{(g_1,\ldots,g_{k},
                 x, g_{k+1}(\theta+1), \ldots,  g_{i-1}(\theta+1)) \mid
                 k\leq i-1, g_j \in A_1^{(0)} \text{ for all } k<j\leq i-1\}$
                 \STATE \textit{rightpart} $:=$ $\{(g_{i+1}(\theta+1),\ldots,
                 g_{k-1}(\theta+1), \partial , g_{k} ,\ldots, g_n)
                 \mid k \geq i+1, g_j \in A_1^{(0)} \text{ for all }
                 i+1 \leq j < k\}$
                 \STATE Append each element in $\{(l_1, \ldots, l_j,
                 r_1 \ldots r_k) \mid (l_1, \ldots, l_j) \in
                 \textit{leftpart}, (r_1, \ldots, r_k) \in
                 \textit{rightpart} $ for $ j,k \in \NN\}$ to \textit{result}.
               \ENDIF
               \IF{$g_i = \theta + 1$ (resp. $g_i = \theta+\frac{1}{q}$)}
                 \STATE $g_i := \partial,x$
                 \STATE \textit{leftpart} $:=$ $\{(g_1,\ldots,g_{k},
                 \partial, g_{k+1}(\theta-1), \ldots,  g_{i-1}(\theta-1)) \mid
                 k\leq i-1, g_j \in A_1^{(0)} \text{ for all } k<j\leq i-1\}$
                 \STATE \textit{rightpart} $:=$ $\{(g_{i+1}(\theta-1),\ldots,
                 g_{k-1}(\theta-1), x , g_{k} ,\ldots, g_n)
                 \mid k \geq i+1, g_j \in A_1^{(0)} \text{ for all }
                 i+1 \leq j < k\}$
                 \STATE Append each element in $\{(l_1, \ldots, l_j,
                 r_1 \ldots r_k) \mid (l_1, \ldots, l_j) \in
                 \textit{leftpart}, (r_1, \ldots, r_k) \in
                 \textit{rightpart}$ for $ j, k \in \NN \}$ to \textit{result}.
               \ENDIF
            \ENDFOR
          \ENDFOR
          \RETURN{\textit{result}}
	\end{algorithmic}
\end{algorithm}

In order to discuss the complexity of finding all factorizations of a
homogeneous element in $A_1$ resp. $Q_1$, we need an upper bound on
the number of possible factorizations.

\begin{lemma}
\label{LemmaHomogFctrIsFinite}
Let $p = p_0\cdot \varphi^{k}$ be a homogeneous polynomial in $A_1$ resp. $Q_1$, where $k\in\NN$, $p_0 \in \KK[\theta]$ and $\varphi \in \{x,\partial\}$ . Furthermore let $n:=\deg_\theta(p_0)$. Then the number of
different factorizations of $p$ is at most
    $$n \cdot n!\cdot \binom{n+k}{k}.$$
\end{lemma}

\begin{proof}
Let us assume that $p_0$ decomposes in $\KK[\theta]$ into $\tilde n \in \NN$
factors, where $\tilde n \leq n$. As all of these factors commute, there are up to $\tilde n!$ different
possibilities to rearrange them. For every such arrangement of the
factors of $p_0$, we can place the $k$ available $\varphi$ at any
position (with applied shift to the respective factors of $p_0$),
which leads to $\binom{\tilde n+k}{k}$ possibilities each time. Finally, due
to Lemma \ref{lem:thetairred}, the linear factors of $p_0$ might split
into $x\partial$ resp. $\partial x$. This would add for each instance
at most $\tilde n$ new distinct factorizations.
As $p_0$ factors at most into linear factors, we can assume $\tilde n=
n$ and obtain the stated upper bound.
\end{proof}


\begin{remark}
\label{RemOnFFD}
In (\cite{BHL14}) we prove that in the case of the polynomial $n$th ($q$-)Weyl algebra, a nonzero polynomial has only finitely many different factorizations. In yet another recent paper (\cite{giesbrecht2014factoring}) we have developed an algorithm for computing all factorizations of a given polynomial in the $n$th ($q$-)Weyl algebra.
\end{remark}

The termination of Algorithms \ref{alg:homogfac} and \ref{alg:homogfac_all} is clear, as we only
iterate over finite sets. The correctness follows by our preliminary
work.


\begin{corollary}
    Given the denotations as in Corollary \ref{cor:complOneFact} By Lemma \ref{LemmaHomogFctrIsFinite}, the number of different factorizations of $p$ is bounded by
    $$n \cdot  n!\cdot \binom{n+|k|}{|k|}.$$

    In order to obtain all these different factorizations, it would require
    \begin{align*}O\Biggl( &n\cdot \rho_q(n,\lceil \log |n!|\rceil) +
    \omega_q(n,b + \lceil \log |n!| \rceil) \\&+
      \left(n^2 + n\cdot n!\cdot
        \binom{n+|k|}{|k|}\right)\mathcal{S}_q(n,1,\lceil \log |b\cdot
        n!|\rceil, \sigma)\Biggr)\end{align*}
    bit operations, where $\sigma(x) = x+1$ if $q=1$, and $\sigma(x) =
    q\cdot x
    + 1$ otherwise. 
    
  \end{corollary}

\subsection{Application to the Rational First Weyl Algebra}

In practice, one is often interested in ordinary differential equations 
over the field of rational functions in the indeterminate $x$.
We refer to the corresponding algebra of operators  as the \textbf{first rational Weyl algebra} and denote it as $B_1$. 
The commutation rules over $B_1$ are extended from those in $A_1$, that is 
$\partial g(x) = g(x) \partial + \tfrac{\partial g(x)}{\partial x}$
for $g(x)\in \KK(x)$. 

Unlike in the polynomial Weyl algebra, an infinite number of
nontrivial factorizations of an element is possible. The
easiest example is the polynomial $\partial^2 \in A_1$, 
having except $\partial\cdot \partial$ a family of nontrivial factorizations \mbox{$(\partial+ \frac{1}{x+c})(\partial - \frac{1}{x+c})$} for all $c \in \KK$ over $B_1$; the only factorization in $A_1$ is $\partial\cdot \partial$.  Thus, at first glance, the factorization problem in both
the rational and the polynomial Weyl algebras seems to be distinct in
general. But there are still many things in common.

The formalism of the \textbf{Ore localization} of a ring (cf. e.~g.~ \cite{Bueso:2003}) can be briefly recalled as follows.
Let $R$ be a domain and $\{0\} \subsetneq S\subset R$ be a multiplicatively closed {\bf Ore set} in $R$, i.~e. the {\bf Ore condition} holds for $S$ and $R$ (the condition will appear below). Then there exists a localized ring, denoted by $S^{-1}R$ together with 
the classical embedding $\iota: R \to S^{-1}R, r \mapsto 1^{-1} r$, 
such that $\iota(S) \subset S^{-1}R$ becomes invertible.
Note, that the presentation of a left fraction $s^{-1}r \in S^{-1}R$
via the tuple $(s,r)\in S\times R$ is by no means unique, but defines an equivalence class. 

Rational Weyl algebras can be recognized as Ore localizations of polynomial Weyl algebras with respect to the multiplicatively closed set $S:=\KK[x]\setminus\{0\}$, which can be proven to be an Ore set in $A_1$. Let us clarify the connection between factorizations in an algebra and in its Ore localization. 

\begin{lemma}
\label{thrm:liftingrationalfact}
Let $R$ be a domain and $S\subset R$ 
be an Ore set in $R$. 
Moreover, let $h$ be an element in $S^{-1}R\setminus\{0\}$. Suppose, that 
$h = h_1 \cdots h_m$, $m \in \NN$, $h_i \in S^{-1}R$ for $i \in \underline{m}$. Then there exists $q\in S$ and $\tilde h_1, \ldots, \tilde h_m \in R$, such that $qh = \tilde h_1\cdots \tilde h_m.$
\end{lemma}
\begin{proof}
Suppose that $h = h_1 h_2 = (s_1^{-1} r_1) \cdot (s_2^{-1} r_2)$ for $r_i \in R, s_i \in S$. Then by the Ore condition $\exists \hat{r}_1 \in R, \hat{s}_2 \in S$ such that $r_1 s_2^{-1} =  \hat{s}_2^{-1} \hat{r}_1$. Thus 
$h = s_1^{-1} \hat{s}_2^{-1} \hat{r}_1 r_2$ and for $q=\hat{s}_2 s_1 \in S$ 
and $\tilde{h}_1 = \hat{r}_1, \tilde{h}_2 = r_2 \in R$ one has 
$q h = \tilde{h}_1 \tilde{h}_2 \in R$. The rest follows by induction.
\end{proof}

Thus we can lift any factorization from the ring
$S^{-1}R$ to a factorization in $R$ by a left multiplication with an
element of $S$. 

\begin{example}
As it was mentioned before, in the first rational Weyl algebra one has
$\partial^2 = (\partial+ \frac{1}{x+c})(\partial - \frac{1}{x+c})$ for all $c \in \KK$. Let us fix $c$ and analyze the lifting. 
\[
(\partial+ (x+c)^{-1})(\partial - (x+c)^{-1}) = 
(x+c)^{-1} \cdot ((x+c)\partial+ 1) \cdot (x+c)^{-1} \cdot ((x+c)\partial - 1)
\]
Since $\partial \cdot (x+c) = (x+c)\partial+ 1$, one has $((x+c)\partial+ 1) \cdot (x+c)^{-1} = \partial$ and thus
\[
\partial^2 = (x+c)^{-1} \cdot  \partial  \cdot ((x+c)\partial - 1),
\]
from which we read off the corresponding factorization in the
polynomial first Weyl algebra
\[
(x+c) \cdot \partial^2 = \partial  \cdot  ((x+c)\partial - 1).
\]
In the notation of the preceding Lemma $q = x+c, \tilde{h}_1 = \partial, \tilde{h}_2 = ((x+c)\partial - 1)$. 
In particular, the infinite family of factorizations we started with does not propagate to the polynomial case: as we see, the parameter $c$ is present in the lifted polynomial $(x+c) \partial^2$. By our approach we can prove, that 
$(x+c) \cdot \partial^2 = \partial  \cdot  ((x+c)\partial - 1)$ are the only
factorizations of $x\partial^2+c\partial^2$ in $A_1$ for any $c \in \KK$.
\end{example}

\begin{proposition}
Let $U:= \{ r\in R \mid 1^{-1}r \in S^{-1}R$ is invertible $\} \subset R$. Then
\noindent 
1. $r\in U$ 
$\Leftrightarrow$ $\exists w\in R : wr\in S$.\\
2. If $S=\KK[x]\setminus\{0\}$ in $R\in \{A_1, Q_1\}$, to any factorization 
of a fraction $h\in S^{-1}R$ we can associate a factorization of $qh \in R$ into elements of $R$.\\
3. Let $1^{-1}r$  be an irreducible element in $S^{-1}R$. Then in any factorization $r = pq$, where $p,q\in R$ one has $p \in U$ or $q\in U$, 
i.~e~ in general $r$ is not irreducible in $R$.\\
4. If $r\in R$ is irreducible in $R$, in general $1^{-1}r$ is not irreducible in $S^{-1}R$.
\end{proposition}


Surprisingly, 
irreducible $[-1,1]$-homogeneous polynomials remain irreducible in the
 rational Weyl algebra, as the following Theorem shows.

\begin{theorem}
\label{thrm:irreducibleRationalWeyl}
Let $p$ be an irreducible $[-1,1]$-homogeneous polynomial in $A_1$. 
Then, in the first rational Weyl algebra $B_1$, $1^{-1}p$ is irreducible up to an invertible multiple.
\end{theorem}
\begin{proof}
The following monic homogeneous polynomials are irreducible in $A_1$:
\begin{enumerate}
\item $\partial$, which is also irreducible over $B_1$,
\item $x$, which is a unit in $B_1$,
\item a monic irreducible $p$ over $\KK[\theta]$, $p\notin\{\theta, \theta+1\}$.
\end{enumerate}
Therefore, the only interesting case is the third one. Now let $p$ be
a monic irreducible element in $A_1^{(0)} \setminus \{\theta, \theta+1\}$. From now on we identify $p$ with $1^{-1}p \in B_1$. 
Suppose, that $p \in F$ is nontrivially reducible over $B_1$, say $p = p_1\cdot p_2$ for $p_1, p_2 \in B_1 \setminus A_1$, both non-invertible, thus $\deg_{\partial}(p_1), \deg_{\partial}(p_2) \geq 1$ and therefore $\deg_{\partial}(p) \geq 2$.

By Lemma \ref{thrm:liftingrationalfact}, there exist $q \in \mathbb{K}[x]$, $\tilde p_1, \tilde p_2 \in A_1 \setminus  \mathbb{K}[x]$, such that $qp = \tilde p_1 \tilde p_2$.

\textbf{Case 1: $q = x^k$, $k \in \mathbb{N}$ (homogeneous
  attempt).}\\
Then all possible factorizations of $x^k \cdot p$ in $A_1$ are due to
Lemma \ref{lem:rewriteKTheta} of the form
$$x^{k-l}p(\theta -l)x^l, l\in \mathbb{N}_0, l \leq k.$$
As shifts of irreducible elements in a univariate commutative
polynomial ring $\KK[\theta]$ are irreducible (see e.~g. \cite{beachy2006abstract}, Section 4.2) and $\deg_{\partial}(p) \geq 2$, 
 we see that $\tilde p_1$ and $\tilde p_2$ as supposed above do not exist.

\textbf{Case 2: $q = \sum_{i=0}^n q_i x^i$, $n \geq 1$, $q_n \neq 0$; $q$ is not a
  single term:}\\
Note, that the product $qp$ in this case is not homogeneous with respect to the $[-1,1]$-grading. Let $m\in\NN, m< n$ be minimal, satisfying $q_m \neq 0$, then the sum in $qp = \sum_{i=m}^n q_i x^i p$ coincides with the graded decomposition of $qp$.

With notations from the proof of Lemma \ref{GradedFctr}, suppose 
that $\alpha(\tilde p_1) = \eta \in \ZZ$ and 
$\alpha(\tilde p_2) = \mu \in \ZZ$. Then
%
\[
q_m x^m p = (qp)_{\alpha(qp)} = (\tilde{p_1} \tilde p_2)_{\alpha(\tilde p_1 \tilde p_2)} =  (\tilde p_1)_{\eta} (\tilde p_2)_{\mu}.\] 
Since $q_m \neq 0$, we can proceed like in  Case 1,
where two kinds of factorization are possible. Let us first 
write $(\tilde p_1)_{\eta} = x^{m-l} p(\theta -l)$ for some $0 \leq l\leq m$
and $(\tilde p_2)_{\mu} = q_m x^{l}$, then $\deg_\partial(\tilde p_1) \geq \deg_\partial(\tilde p_1)_{\eta}= \deg_\partial(p) = \deg_\partial(qp) =
\deg_\partial(\tilde p_1 \tilde p_2) = 
\deg_\partial(\tilde p_1) + \deg_\partial(\tilde p_2)$, indicating that 
$\deg_\partial(\tilde p_2) = 0$ and $\deg_\partial(\tilde p_1) =\deg_\partial(p)$. That is, $\tilde p_2$ must be in $\KK[x]$ and therefore cannot be as supposed above.
The second case, where $\deg_\partial (\tilde p_2)_{\mu} = \deg_\partial(p)$
is analogous and thus the proof is completed.




\end{proof}


\section{Implementation and benchmarking}

We implemented the presented algorithms in \textsc{Singular}, and
since version 3-1-3 they are part of the distribution of
\textsc{Singular}. The following example shows how to use the library
containing them.

\begin{example}
  Let $h\in Q_1$ be the polynomial
\begin{eqnarray*}
h 
   &  := &q^{25}x^{10}\partial^{10}+q^{16}(q^{4}+q^{3}+q^{2}+q+1)^{2}x^{9}\partial^{9}\\
   &&  +q^{9}(q^{13}+3q^{12}+7q^{11}+13q^{10}+20q^{9}+26q^{8}\\
   && +30q^{7}+31q^{6}+26q^{5}+20q^{4}+13q^{3}+7q^{2}+3q+1)x^{8}\partial^{8}\\
   && +q^{4}(q^{9}+2q^{8}+4q^{7}+6q^{6}+7q^{5}+8q^{4}+6q^{3}+4q^{2}+2q+1)\\
   && (q^{4}+q^{3}+q^{2}+q+1)(q^{2}+q+1)x^{7}\partial^{7}\\
   &&+q(q^{2}+q+1)(q^{5}+2q^{4}+2q^{3}+3q^{2}+2q+1)\\
   &&(q^{4}+q^{3}+q^{2}+q+1)(q^{2}+1)(q+1)x^{6}\partial^{6}\\
   &&+(q^{10}+5q^{9}+12q^{8}+21q^{7}+29q^{6}+33q^{5}\\
   &&+31q^{4}+24q^{3}+15q^{2}+7q+12)x^{5}\partial^{5}+6x^{3}\partial^{3}+24
\end{eqnarray*}
and $\KK = \QQ$.
We can use \textsc{Singular} to obtain all of its factorizations in the
following way.
\begin{verbatim}
LIB "ncfactor.lib";
ring R = (0,q),(x,d),dp;
def r = nc_algebra (q,1);
setring(r);
poly h = ... //See the polynomial defined above.
homogfacFirstQWeyl_all(h);
[1]:
   [1]:
      1
   [2]:
      x5d5+x3d3+4
   [3]:
      x5d5+6
[2]:
   [1]:
      1
   [2]:
      x5d5+6
   [3]:
      x5d5+x3d3+4
\end{verbatim}

As one can see here, the output is a list containing lists containing
elements in $Q_1$. Those elements in $Q_1$ are factors of $h$, and each list
represents one possible factorization of $h$.

If the user is interested in just one factorization the command
\texttt{homogfacFirstQWeyl} instead of \texttt{homogfacFirstQWeyl\_all}
can be used. The output will then be just one list containing elements
in $Q_1$.

The calculation was run on a on a computer
with a 4-core Intel CPU (Intel{\textregistered}   Core\texttrademark
 i7-3520M CPU with 2.90GHz, 2 physical cores, 2 hardware threads, 32K
L1[i,d], 256K L2, 4MB L3 cache), 16GB RAM and Ubuntu 12.04LTS as
operating system.
The computation time was 0.62 seconds. 
\end{example}
\begin{remark}
  The factorization of products of homogeneous elements in $A_1$ can
  be observed to be faster than the factorization of the same products
  in $Q_1$. The element in the example above,
  i.e. $(x^5\partial^5+6)(x^5\partial^5+x^3\partial^3+4),$
  viewed as an element in $A_1$, takes 0.08s to factorize compared to
  0.62s in the $q$-Weyl case.
This seems to be
way slower considering that both algorithms have the same complexity. But this slowdown is not
due to more steps that need to be done in the algorithm for the
$q$-Weyl algebra, but due to
the parameter $q$ and the speed of calculating in $\QQ(q)$ as the basefield instead of just
in $\QQ$. The internal arithmetic in \textsc{Singular} to handle
parametrized basefields is being improved by the \textsc{Singular} team.
\end{remark}

In fact, there is no computer algebra system known to the authors that
can factor polynomials in the first $q$-Weyl algebra $Q_1$. Therefore,
we cannot compare our algorithms in this case to other implementations.

For the first Weyl algebra $A_1$, there exist other implementations. We can draw a
comparison to the \texttt{DFactor} method in the \texttt{DETools}
package of \textsc{Maple} and the \texttt{nc\_factorize\_all} method
in the \texttt{NCPoly} library of \textsc{REDUCE}.
Furthermore, we were provided with a wrapper for the algorithm ``Coprime
Index 1 Factorizations'' (\texttt{CP1F}) mentioned in \cite{van1997formal} dealing
with polynomials of the form $\KK[x][\theta]$ in order to be able to
compare it to the algorithm for this special case explicitly. This
guarantees a fair evaluation on a core level for an intersection with
homogeneous polynomials that does not invoke the complete
factorization machinery implemented in \texttt{DFactor}.

In the next subsection, we will only compare \texttt{DFactor} and
\texttt{nc\_factorize\_all} to our implementation. Later on, we will
compare the wrapper of \texttt{CP1F} implemented in \textsc{Maple} to
our implementation, as we have to choose for the comparison a special set of polynomials,
namely the homogeneous ones supported by \texttt{CP1F}.

\subsection{Comparison to \texttt{DFactor} and \texttt{nc\_factorize\_all}}
We used version 17 of
\textsc{Maple} and version 3.8 of \textsc{REDUCE}.
In order to make our benchmarks reproducible, we utilized the
\textsc{SDEval} framework presented in \cite{HLN2013}. You can download
the sources and the results of the computations on one of the author's
website: \url{https://cs.uwaterloo.ca/~aheinle/software_projects.html}.

\begin{remark}
  As mentioned before, the algorithm \texttt{DFactor} implemented in \textsc{Maple}
  factorizes over the rational Weyl algebra, i.e. the variable $x$ is
  a rational argument having adjusted commutation rules with
  $\partial$. This is a weaker assumption on the input since the ring
  that is dealt with there is larger. The comparison is still valid,
  since we have shown in Theorem \ref{thrm:irreducibleRationalWeyl} that a
  factorization of a homogeneous polynomial into irreducible elements
  over $A_1$ cannot be further refined in the first Weyl algebra with
  rational coefficients.
  
  We will not go into detail about how the algorithm in \textsc{Maple}
  works. The interested
  reader can find details in \cite{Hoeij:1997}. It works with
  collections of exponential parts and their multiplicities at all
  singularities of a given differential operator $f$ and subsequent
  calculation of left and right hand factors.

  The algorithm implemented in \textsc{REDUCE} is also working with
  the polynomial Weyl algebra. In fact, the algorithm written there
  can be applied to a broad class of polynomial noncommutative
  rings.

  Details about the functionality of the algorithm in \textsc{REDUCE}
  are unfortunately not available. One can only try to understand it
  from the code that is given open source. It uses several Gr\"obner
  basis computations in order to find its solutions.
\end{remark}

\begin{example}
  Consider again the element $$h :=
  (x^5\partial^5+6)\cdot(x^5\partial^5+x^3\partial^3+4)$$
  in expanded form.
  \begin{itemize}
    \item \textsc{Singular}: Found two factorizations in less than a
      second.
    \item \textsc{Maple}: Found one factorization after 29 seconds;
      The factors are huge (size of the output file is around 100KB).
    \item \textsc{REDUCE}: Did not terminate after 9 hours of calculation.
  \end{itemize}
\end{example}

\begin{example}
  We experimented with other randomly generated products of two homogeneous polynomials in
  the first Weyl algebra. The results are listed in the next table. An
  entry labeled with ``-- NT --'' stands for ``no termination after
  two hours''.
  \newpage
  \begin{center}
  \begin{longtable}{ | p{130pt} | p{130pt} | p{130pt} |}
    \hline
    \textbf{\textsc{Singular}} & \textbf{\textsc{Maple}} &
    \textbf{\textsc{REDUCE}}\\
    \hline
    \hline
    \multicolumn{3}{| c |}{$(x^{10}\partial^{10}+5x\partial +
      7) \cdot x^2 \cdot (x^{11}\partial^{11}+3x^7\partial^7+x\partial+4)$:}\\
    \hline
    0.08s; 12 factorizations & -- NT -- & SEGMENTATION FAULT\\ 
    \hline
    \hline
    \multicolumn{3}{| c
      |}{$(x^5\partial^5+6)\cdot(x^5\partial^5+x^3\partial^3+4) \cdot \partial^{10}$:}\\
    \hline
    0.77s; 132 factorizations & 11.18s; 1 factorizations & -- NT -- \\
    \hline
    \hline
    \multicolumn{3}{| c |}{$(5x^{10}\partial^{10} + 7x^9\partial^9 + 8x^8\partial^8 + 9x^7\partial^7 + 6x^6\partial^6 + 5x^5\partial^5 + 8x^4\partial^4 + 5x^3\partial^3 + 9x^2\partial^2 + 9x\partial + 6)\cdot\partial^{20}$:}\\
    \hline
    0.18s; 21 factorizations & -- NT -- & -- NT --\\
    \hline
    \hline
    \multicolumn{3}{| c |}{$(7x^{15}\partial^{15} +
      x^{13}\partial^{13} - x^{12}\partial^{12} - 3x^{10}\partial^{10}
      + 2x^9\partial^9 + x^8\partial^8 + x^7\partial^7 - x^5\partial^5
      - 9x^4\partial^4 + x\partial - 1)\cdot$}\\
    \multicolumn{3}{| c |}{$ (8x^{13}\partial^{13} +
      3x^{12}\partial^{12} + x^{11}\partial^{11} -
      2x^{10}\partial^{10} + 10x^8\partial^8 - 3x^7\partial^7 +
      2x^5\partial^5 + x^4\partial^4 + 38x\partial + 1)\cdot \partial^6$:}\\
    \hline
    5.88s; 504 factorizations & -- NT -- & -- NT -- \\
    \hline
    \hline
    \multicolumn{3}{| c |}{$(x^{10}\partial^{10} + 23x^9\partial^9 +
      3x^8\partial^8 - 9x^7\partial^7 - x^5\partial^5 + 3x^4\partial^4
      + 6x^3\partial^3 + 4x\partial + 1) \cdot $}\\
    \multicolumn{3}{| c |}{$ (-x^8\partial^8 + 4x^7\partial^7 -
      x^6\partial^6 + 4x^5\partial^5 - 5x^4\partial^4 + x^2\partial^2
      - 7x\partial - 10) \cdot x^{10}$:}\\
    \hline
    0.76s; 132 factorizations & -- NT -- & -- NT -- \\
    \hline
    \hline    
    \multicolumn{3}{| c |}{$(-2x^{24}\partial^{24} + x^{23}\partial^{23} + 4x^{22}\partial^{22} -
      110x^{21}\partial^{21} + x^{20}\partial^{20} + x^{19}\partial^{19} + x^{18}\partial^{18} +
      x^{17}\partial^{17} +$}\\
    \multicolumn{3}{| c |}{$ 5x^{16}\partial^{16} - 7x^{15}\partial^{15} + 4x^{14}\partial^{14} - x^{13}\partial^{13} + x^{12}\partial^{12}
      - 2x^{11}\partial^{11} + x^9\partial^9 + 5x^8\partial^8 + x^7\partial^7 + $}\\
    \multicolumn{3}{| c |}{$6x^5\partial^5 +
      x^4\partial^4 + 2x^3\partial^3 + 219x^2\partial^2 + x\partial - 1) \cdot(-x^{25}\partial^{25} + x^{24}\partial^{24} -
      32x^{23}\partial^{23} + x^{22}\partial^{22} + $}\\
    \multicolumn{3}{| c |}{$7x^{21}\partial^{21} + 61x^{20}\partial^{20} -
      2x^{18}\partial^{18} + x^{16}\partial^{16} + 2x^{15}\partial^{15}- 2x^{14}\partial^{14} -$}\\
    \multicolumn{3}{| c |}{$x^{12}\partial^{12} - 3x^{11}\partial^{11} + 2x^{10}\partial^{10} + 2x^8\partial^8 - 9x^7\partial^7
      - x^6\partial^6 + x^5\partial^5 + 4x^3\partial^3 + x^2\partial^2)$:}\\
    \hline
    28.23s; 230 factorizations & -- NT -- & -- NT -- \\
    \hline
    \hline
    \multicolumn{3}{| c |}{$(x^{10}\partial^{10} + 13x^9\partial^9 - x^8\partial^8 +
      4x^7\partial^7 + 13x^6\partial^6 - 3x^5\partial^5 - 37x^4\partial^4 - x^3\partial^3 +
      x^2\partial^2 + x\partial - 1)\cdot $}\\
    \multicolumn{3}{| c |}{$(-x^{10}\partial^{10} -23x^9\partial^9 +
      3x^8\partial^8 + x^7\partial^7 - x^6\partial^6 - 2x^5\partial^5
      - 2x^4\partial^4 + 2x^3\partial^3 - x^2\partial^2 - 2x\partial -
      2)$:}\\
    \hline
    0.06s; 6 factorizations & -- NT -- & -- NT --\\
    \hline
    \hline
    \multicolumn{3}{| c |}{$(98x^{15}\partial^{15} +
      40x^{14}\partial^{14} + 98x^{13}\partial^{13} +
      44x^{12}\partial^{12} + 55x^{11}\partial^{11} +
      96x^{10}\partial^{10} + 95x^9\partial^9 +$}\\
    \multicolumn{3}{| c |} {$7x^8\partial^8 + 56x^7\partial^7 +
      56x^6\partial^6 + 40x^5\partial^5 + 11x^4\partial^4 +
      40x^3\partial^3 + 78x^2\partial^2 +$}\\
    \multicolumn{3}{| c |} {$13x\partial + 19) \cdot (61x^{15}\partial^{15} +
      50x^{14}\partial^{14} + 83x^{13}\partial^{13} +
      11x^{12}\partial^{12} + 89x^{11}\partial^{11} +$}\\
    \multicolumn{3}{| c |} {$55x^{10}\partial^{10} + 81x^9\partial^9 +
      63x^8\partial^8 + 22x^7\partial^7 + 10x^6\partial^6 +$}\\
    \multicolumn{3}{| c |} {$35x^5\partial^5 + 90x^4\partial^4 +
      60x^3\partial^3 + 20x^2\partial^2 + 30x\partial + 43)$:}\\
    \hline
    0.08s; 2 factorizations & --NT -- & -- NT -- \\
    \hline
    \hline
     \multicolumn{3}{| c |}{$ (85x^{20}\partial^{20} +
       80x^{19}\partial^{19} + 27x^{18}\partial^{18} +
       74x^{17}\partial^{17} + 49x^{16}\partial^{16} +
       95x^{15}\partial^{15} + 96x^{14}\partial^{14}$}\\
     \multicolumn{3}{| c |}{$+ 37x^{13}\partial^{13} +
       26x^{12}\partial^{12} + 93x^{11}\partial^{11} +
       39x^{10}\partial^{10} + 19x^9\partial^9 + 48x^8\partial^8 +
       82x^7\partial^7$}\\
     \multicolumn{3}{| c |}{$+ 26x^6\partial^6 + 26x^5\partial^5 +
       7x^4\partial^4 + 61x^3\partial^3 + 8x^2\partial^2 + 81x\partial
       + 88)^2$:}\\
     \hline
     0.08s; 1 factorizations & -- NT -- & -- NT --\\
     \hline
  \end{longtable}
  \end{center}
\end{example}

The conclusion we can draw at this point is: Even if homogeneous polynomials seem to
be easy objects to factorize according to the algorithm we propose, they seem to form a worst case class for
the implementations in \textsc{REDUCE} and \textsc{Maple}.

Therefore, with our algorithm we are now able to factorize more
polynomials using computer algebra systems: homogeneous
polynomials in $Q_1$ in general, and for $A_1$ we have broaden the range of
polynomials that can be factorized in a feasible amount of time or
even sometimes at all.

Moreover, our approach can be used to
enhance existing algorithms and their implementations. Namely, since
checking a given polynomial for the homogeneity is a very cheap
procedure as we have seen in Lemma \ref{lem:cheapHomogTest}, and since
for the case of a homogeneous polynomial our proposed algorithm can be applied,
the algorithm for the case of homogeneous polynomials --
appearing, for instance, as factors of a bigger polynomial --
can be eliminated from further computations.

\subsection{Comparison to \texttt{CP1F}}

As indicated before, we were provided a wrapper to the function
implemented in \textsc{Maple} that represents \texttt{CP1F}, whose supported input
polynomials are of the form $\KK[x][\theta]$. Hence, there is a
nontrivial intersection with homogeneous polynomials in
$A_1$. Comparing it to the implementation of our Algorithm
\ref{alg:homogfac} on homogeneous polynomials of $\theta$-degree
between 20 and 400, we obtain the following timings.

\begin{center}
  \begin{longtable}{ | p{100pt} | p{100pt} | p{100pt} |}
    \hline
    \textbf{Example} & \textbf{Algorithm \ref{alg:homogfac}} & \textbf{\texttt{CP1F}}\\
    \hline
    \hline
    Degree 20 & 0.04s & 0.17s \\
    \hline
    Degree 40 & 0.07s & 0.61s \\
    \hline
    Degree 60 & 0.11s & 1.66s\\
    \hline
    Degree 100 & 0.26s & 6.39s\\
    \hline
    Degree 200 & 2.03s & 296.78s\\
    \hline
    Degree 250 & 2.86s & 454.17s\\
    \hline
    Degree 300 & 5.9s  & 370.49s\\
    \hline
    Degree 350 & 8.78s  & 1741.53s\\
    \hline
    Degree 400 & 14.62s & 4355.32s\\
    \hline
  \end{longtable}
\end{center}

We can derive from this table that for small degrees, the timings are close
to each other. 
With increasing degree though, the difference in performance becomes more visible, and one observes also different asymptotic behaviours, as Figure \ref{fig:graphic} visualizes.

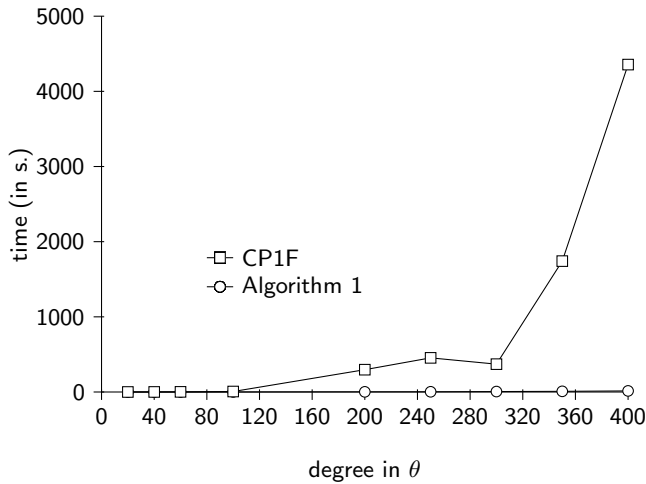
\begin{figure}
\caption{Visualization of asymptotic behaviour of \texttt{CP1F} and
  Algorithm \ref{alg:homogfac}}
\label{fig:graphic}

\begin{tikzpicture}[y=.2cm, x=.7cm,font=\sffamily]
	\draw (0,0) -- coordinate (x axis mid) (10,0);
    	\draw (0,0) -- coordinate (y axis mid) (0,25);
    	\foreach \x/\xtext in {0/0,1/40,2/80,3/120,4/160,5/200,6/240,7/280,8/320,9/360,10/400}
     		\draw (\x,1pt) -- (\x,-3pt)
			node[anchor=north] {\xtext};
    	\foreach \y/\ytext in {0/0,5/1000,10/2000,15/3000,20/4000,25/5000}
     		\draw (1pt,\y) -- (-3pt,\y) 
     			node[anchor=east] {\ytext}; 
	\node[below=0.8cm] at (x axis mid) {degree in $\theta$};
	\node[rotate=90, above=0.8cm] at (y axis mid) {time (in s.)};

	\draw plot[mark=*, mark options={fill=white}] 
		file {homog.data};
	\draw plot[mark=square*, mark options={fill=white}]
		file {cp1f.data};
	
	\begin{scope}[shift={(2,7)}] 
	\draw (0,0) -- 
		plot[mark=*, mark options={fill=white}] (0.25,0) -- (0.5,0) 
		node[right]{Algorithm \ref{alg:homogfac}};
	\draw[yshift=\baselineskip] (0,0) -- 
		plot[mark=square*, mark options={fill=white}] (0.25,0) -- (0.5,0)
		node[right]{CP1F};
	\end{scope}
\end{tikzpicture}
\end{figure}


\section{Conclusion and Future Work}

With this paper, we contributed an algorithm for the
factorization problem considering $[-1,1]$-homogeneous polynomials in the first $q$-Weyl
algebra over an arbitrary field $\KK$. For computable
fields, we discussed a complexity estimate for our
approach. Our approach is implemented and distributed with the
computer algebra system \textsc{Singular}.

Furthermore, we also considered the special case of the first Weyl
algebra and showed that our algorithm beats for the large class of $[-1,1]$-
homogeneous polynomials current implementations in terms of speed and
elegancy of the solutions. Due to Theorem
\ref{thrm:irreducibleRationalWeyl}, we can even state that the
factorizations that our algorithm finds cannot be further refined when
factoring over the rational Weyl algebra. This result is interesting
by itself and could play a role for future research on the question
how to characterize arbitrary irreducible elements in the polynomial
first Weyl algebra, that become reducible after localization.

We can  construct a family of
polynomials where the implementation in \textsc{Singular} is the only
one that is able to factorize those elements in a feasible amount of
time and memory consumption. As our techniques are easy to implement,
they can be used to extend existing implementations in order to broaden
the range of polynomials in the first Weyl algebra that we are
nowadays able to factorize using a computer algebra system.

The canonical next step would be to factor general polynomials in the first
\mbox{($q$-)Weyl} algebra. A first attempt to that was done in
\cite{Heinle:2010}. We made highly use of our knowledge about the
grading of the first Weyl algebra. The approach has been almost
completely of combinatorial nature. Its speed and quality of solutions
was comparable in many cases to the other implementations, but there
was still space for improvement. This improvement has been made in
\cite{Heinle:2012}. \texttt{ncfactor.lib} is distributed with \textsc{Singular} since
version 3-1-3, and the improved version is available since version
3-1-6.
Recently, we generalized our ideas from this paper to the $n$th polynomial Weyl
algebra, where we are now able to factorize general polynomials, as
one can see in \cite{giesbrecht2014factoring}.

As another future work, we plan to generalize Theorem
\ref{thrm:irreducibleRationalWeyl} to homogeneous polynomials in the
$n$th Weyl algebra. Moreover, a generalization of our complexity
estimates for factoring homogeneous polynomials in the first
($q$-)Weyl algebra to the $n$th ($\underline{q}$)-Weyl algebra is planned.


\section*{Acknowledgements}

We would like to thank to Dima Grigoriev for discussions on the subject, 
and Mark van Hoeij for his expert opinion. Many thanks to Mark Giesbrecht for his helpful suggestions and comments. We are grateful to
Wolfram Koepf and Martin Lee for providing us with interesting
examples, and to Daniel Rettstadt and Johannes Hoffmann for stimulating
exchange of opinions.

We acknowledge the helpful suggestions and comments of the anonymous referees.

\renewcommand\bibsection{\section*{References}}

\nobreak

\begin{thebibliography}{0}
\providecommand{\natexlab}[1]{#1}
\providecommand{\url}[1]{\texttt{#1}}
\expandafter\ifx\csname urlstyle\endcsname\relax
  \providecommand{\doi}[1]{doi: #1}\else
  \providecommand{\doi}{doi: \begingroup \urlstyle{rm}\Url}\fi

\end{thebibliography}


\begin{thebibliography}{35}
\providecommand{\natexlab}[1]{#1}
\providecommand{\url}[1]{\texttt{#1}}
\expandafter\ifx\csname urlstyle\endcsname\relax
  \providecommand{\doi}[1]{doi: #1}\else
  \providecommand{\doi}{doi: \begingroup \urlstyle{rm}\Url}\fi

\bibitem[Beachy and Blair(2006)]{beachy2006abstract}
J.~A. Beachy and W.~D. Blair.
\newblock \emph{Abstract algebra}.
\newblock Waveland Press, 2006.

\bibitem[Bell et~al.(2014)]{BHL14}
J.~P. Bell, A. Heinle and V. Levandovskyy.
\newblock On Noncommutative Finite Factorization Domains.
\newblock To appear in the \emph{Transactions of the American
  Mathematical Society}; \emph{ArXiv preprint 1410.6178}, 2014.
\newblock URL \url{http://arxiv.org/abs/1410.6178}

\bibitem[Beals and Kartashova(2005)]{beals2005constructively}
R.~Beals and E.~Kartashova.
\newblock Constructively factoring linear partial differential operators in two
  variables.
\newblock \emph{Theor. Math. Phys.}, 145\penalty0 (2):\penalty0 1511--1524,
  2005.
\newblock URL \url{http://dx.doi.org/10.1007/s11232-005-0178-7}.

\bibitem[Bueso et~al.(2003)]{Bueso:2003}
J.~Bueso, J.~G\'omez-Torrecillas, and A.~Verschoren.
\newblock \emph{{Algorithmic methods in non-commutative algebra. Applications
  to quantum groups.}}
\newblock {Dordrecht: Kluwer Academic Publishers}, 2003.

\bibitem[Chistov(1986)]{chistov1986algorithm}
A.~L. Chistov.
\newblock Algorithm of polynomial complexity for factoring polynomials and
  finding the components of varieties in subexponential time.
\newblock \emph{Journal of Soviet Mathematics}, 34\penalty0 (4):\penalty0
  1838--1882, 1986.

\bibitem[Decker et~al.(2012)]{Singular:2012}
W.~Decker, G.-M. Greuel, G.~Pfister, and H.~Sch{\"o}nemann.
\newblock {\sc Singular} {3-1-6} --- {A} computer algebra system for polynomial
  computations.
\newblock 2012.
\newblock URL \url{http://www.singular.uni-kl.de}.

\bibitem[Foupouagnigni et~al.(2004)]{foupouagnigni2004factorization}
M.~Foupouagnigni, W.~Koepf, and A.~Ronveaux.
\newblock Factorization of fourth-order differential equations for perturbed
  classical orthogonal polynomials.
\newblock \emph{Journal of computational and applied mathematics}, 162\penalty0
  (2):\penalty0 299--326, 2004.

\bibitem[Giesbrecht(1998)]{giesbrecht1998factoring}
M.~Giesbrecht.
\newblock Factoring in skew-polynomial rings over finite fields.
\newblock In \emph{Journal of Symbolic Computation}, 26.4, pages
463--486, 1998, Elsevier.

\bibitem[Giesbrecht et~al.(2014)]{giesbrecht2014factoring}
M.~Giesbrecht, A.~Heinle, and V.~Levandovskyy.
\newblock Factoring linear differential operators in n variables.
\newblock In \emph{Proceedings of the 39th International Symposium on Symbolic
  and Algebraic Computation}, ISSAC '14, pages 194--201, New York, NY, USA,
  2014. ACM.
\newblock URL \url{http://doi.acm.org/10.1145/2608628.2608667}.
\newblock Note: The extended version of the paper has been accepted to 
publication in the \emph{Journal of Symbolic Computation}.

\bibitem[Giesbrecht and Zhang(2003)]{giesbrecht2003factoring}
M.~Giesbrecht and Y.~Zhang.
\newblock Factoring and Decomposing Ore Polynomials over Fq(T).
\newblock In \emph{Proceedings of the 2003 International Symposium on Symbolic
  and Algebraic Computation}, ISSAC '03, pages 127--134, New York, NY, USA,
  2003. ACM.
\newblock URL \url{http://doi.acm.org/10.1145/860854.860888}.

\bibitem[Greuel and Pfister(2007)]{GreuelPfister:2007}
G.-M. Greuel and G.~Pfister.
\newblock \emph{{A Singular introduction to commutative algebra. With
  contributions by Olaf Bachmann, Christoph Lossen and Hans Sch\"onemann. 2nd
  extended ed.}}
\newblock Berlin: Springer, 2007.

\bibitem[Levandovskyy et~al.(2010)]{Plural:2010}
G.-M. Greuel, V.~Levandovskyy, A.~Motsak, and H.~Sch{\"o}nemann.
\newblock {\textsc{Plural}. A \textsc{Singular} 3.1 Subsystem for Computations with Non-commutative Polynomial Algebras. Centre for Computer Algebra, TU Kaiserslautern}, 2010.
\newblock URL \url{http://www.singular.uni-kl.de}.

\bibitem[Grigoriev(1984)]{grigor1984factoring}
D.~Grigoriev.
\newblock Factoring polynomials over a finite field and solving systems of
  algebraic equations.
\newblock \emph{Zapiski Nauchnykh Seminarov POMI}, 137:\penalty0 20--79, 1984.

\bibitem[Grigoriev and Schwarz(2004)]{GrigorievSchwartz:2004}
D.~Grigoriev and F.~Schwarz.
\newblock {Factoring and solving linear partial differential equations.}
\newblock \emph{Computing}, 73\penalty0 (2):\penalty0 179--197, 2004.
\newblock DOI \url{http://dx.doi.org/10.1007/s00607-004-0073-3}.

\bibitem[Heinle(2010)]{Heinle:2010}
A.~Heinle.
\newblock Factorization of polynomials in a class of noncommutative algebras.
\newblock Bachelor Thesis at RWTH Aachen University, April 2010.
\newblock URL \url{https://cs.uwaterloo.ca/~aheinle/bachelorthesis.pdf}

\bibitem[Heinle(2012)]{Heinle:2012}
A.~Heinle.
\newblock Factorization, similarity and matrix normal forms over certain ore
  domains.
\newblock Master's Thesis at RWTH Aachen University, September 2012.
\newblock \url{https://cs.uwaterloo.ca/~aheinle/masterthesis.pdf}

\bibitem[Heinle et al.(2013)]{HLN2013}
A.~Heinle, V.~Levandovskyy and A.~Nareike.
\newblock SymbolicData:SDEval –- Benchmarking for Everyone.
\newblock \emph{ArXiv preprint 1310.5551}, 2013.
\newblock URL \url{http://arxiv.org/abs/1310.5551}

\bibitem[Kac and Cheung(2002)]{Cheung:2002}
V.~Kac and P.~Cheung.
\newblock \emph{{Quantum calculus.}}
\newblock {New York, NY: Springer}, 2002.

\bibitem[Kaltofen(1982)]{Kaltofen:1982}
E.~Kaltofen.
\newblock \emph{On the complexity of factoring polynomials with integer
  coefficients}.
\newblock PhD thesis, Rensselaer Polytechnic Institute, 1982.

\bibitem[Kashiwara(1983)]{Kashiwara:1983}
M.~Kashiwara.
\newblock Vanishing cycle sheaves and holonomic systems of differential
  equations.
\newblock In \emph{Algebraic geometry}, pages 134--142. Springer, 1983.

\bibitem[Lenstra et~al.(1982)]{LLL:1982}
A.~K. Lenstra, H.~W. Lenstra, and L.~Lov{\'a}sz.
\newblock Factoring polynomials with rational coefficients.
\newblock \emph{Mathematische Annalen}, 261\penalty0 (4):\penalty0 515--534,
  1982.

\bibitem[Loewy(1903)]{Loewy:1903}
A.~Loewy.
\newblock {\"Uber reduzible lineare homogene Differentialgleichungen.}
\newblock \emph{Math. Ann.}, 56:\penalty0 549--584, 1903.
\newblock DOI \url{http://dx.doi.org/10.1007/BF01444307}.

\bibitem[Loewy(1906)]{Loewy:1906}
A.~Loewy.
\newblock {\"Uber vollst\"andig reduzible lineare homogene
  Differentialgleichungen.}
\newblock \emph{Math. Ann.}, 62:\penalty0 89--117, 1906.
\newblock DOI \url{http://dx.doi.org/10.1007/BF01448417}.

\bibitem[Malgrange(1983)]{Malgrange:1983}
B.~Malgrange.
\newblock {Polyn\^omes de Bernstein-Sato et cohomologie evanescente.}
\newblock \emph{Ast\'erisque}, 101-102:\penalty0 243--267, 1983.

\bibitem[Melenk and Apel(1994)]{Melenk:1994}
H.~Melenk and J.~Apel.
\newblock \emph{REDUCE package NCPOLY: Computation in non-commutative
  polynomial ideals.}
\newblock Konrad-Zuse-Zentrum Berlin (ZIB), 1994.

\bibitem[Monagan et~al.(2008)]{Maple}
M.~B. Monagan, K.~O. Geddes, K.~M. Heal, G.~Labahn, S.~M. Vorkoetter,
  J.~McCarron, and P.~DeMarco.
\newblock \emph{Maple Introductory Programming Guide}.
\newblock Maplesoft, Waterloo ON, Canada, 2008.

\bibitem[Nikiforov and Uvarov(1988)]{nikiforov1988special}
A.~F. Nikiforov and V.~B. Uvarov.
\newblock Special functions of mathematical physics: a unified introduction
  with applications.
\newblock 1988.

\bibitem[Saito et~al.(2000)]{SaStuTaka:2000}
M.~Saito, B.~Sturmfels, and N.~Takayama.
\newblock \emph{{Gr\"obner deformations of hypergeometric differential
  equations.}}
\newblock {Berlin: Springer}, 2000.

\bibitem[Schwarz(2009)]{Schwarz:2009}
F.~Schwarz.
\newblock Alltypes in the web.
\newblock \emph{ACM Commun. Comput. Algebra}, 42\penalty0 (3):\penalty0
  185--187, February 2009.
\newblock URL \url{http://doi.acm.org/10.1145/1504347.1504379}.

\bibitem[Shemyakova(2007)]{parameters:shemyakova:2007}
E.~Shemyakova.
\newblock Parametric factorizations of second-, third- and fourth-order linear
  partial differential operators with a completely factorable symbol on the
  plane.
\newblock \emph{Mathematics in Computer Science}, 1\penalty0 (2):\penalty0
  225--237, 2007.

\bibitem[Shemyakova(2009)]{shemyakova2009multiple}
E.~Shemyakova.
\newblock Multiple factorizations of bivariate linear partial differential
  operators.
\newblock In \emph{{Proc. CASC 2009}}, pages 299--309. Springer, 2009.

\bibitem[Shemyakova(2010)]{2010:shemyakova:refinement}
E.~Shemyakova.
\newblock Refinement of two-factor factorizations of a linear partial
  differential operator of arbitrary order and dimension.
\newblock \emph{Mathematics in Computer Science}, 4:\penalty0 223--230, 2010.
\newblock URL \url{http://dx.doi.org/10.1007/s11786-010-0052-3}.

\bibitem[Tsarev(1994)]{Tsarev:1994}
S.~P. Tsarev.
\newblock {Problems that appear during factorization of ordinary linear
  differential operators.}
\newblock \emph{Program. Comput. Softw.}, 20\penalty0 (1):\penalty0 27--29,
  1994.

\bibitem[Tsarev(1996)]{Tsarev:1996}
S.~P. Tsarev.
\newblock An algorithm for complete enumeration of all factorizations of a
  linear ordinary differential operator.
\newblock In \emph{Proceedings of the 1996 international symposium on Symbolic and Algebraic Computation},  ISSAC'96, pages 226--231. ACM Press, 1996.

\bibitem[van Hoeij(1996)]{Hoeij:1996}
M.~van Hoeij.
\newblock \emph{Factorization of linear differential operators}.
\newblock Nijmegen, 1996.
\newblock URL \url{http://books.google.de/books?id=rEmjPgAACAAJ}.

\bibitem[van Hoeij(1997{\natexlab{a}})]{Hoeij:1997}
M.~van Hoeij.
\newblock {Factorization of differential operators with rational functions
  coefficients.}
\newblock \emph{J. Symb. Comput.}, 24\penalty0 (5):\penalty0 537--561,
  1997{\natexlab{a}}.
\newblock DOI \url{http://dx.doi.org/10.1006/jsco.1997.0151}.

\bibitem[van Hoeij(1997{\natexlab{b}})]{van1997formal}
M.~van Hoeij.
\newblock Formal solutions and factorization of differential operators with
  power series coefficients.
\newblock \emph{Journal of Symbolic Computation}, 24\penalty0 (1):\penalty0
  1--30, 1997{\natexlab{b}}.
\newblock DOI \url{http://dx.doi.org/10.1006/jsco.1997.0110}.

\bibitem[van Hoeij and Yuan(2010)]{Hoeij:2010}
M.~van Hoeij and Q.~Yuan.
\newblock {Finding all Bessel type solutions for linear differential equations
  with rational function coefficients}.
\newblock In \emph{Proceedings of the 2010 International Symposium on Symbolic and Algebraic Computation}, ISSAC '10, pages 37--44, New York, NY, USA.  ACM Press, 2010.
\newblock DOI \url{http://dx.doi.org/10.1145/1837934.1837948}.

\bibitem[von~zur Gathen and Gerhard(2013)]{von2013modern}
J.~von~zur Gathen and J.~Gerhard.
\newblock \emph{Modern computer algebra}.
\newblock Cambridge University Press, 2013.

 

\end{thebibliography}
\newcommand{\Hoeven}{\relax}


\end{document}